\documentclass[10pt,conference]{IEEEtran}
\usepackage{amsmath,amssymb}
\usepackage{graphicx}
\usepackage{mathrsfs}
\usepackage{tikz}
\usetikzlibrary{arrows}
\tikzset{>=latex'}
\usepackage{booktabs}
\usepackage{nicefrac}
\usepackage{enumerate}
\usepackage{psfrag}	
\usepackage{array}
\usepackage{multirow,hhline}
\usepackage{exscale}
\usepackage{color}
\usepackage{colortbl}
\usepackage{epsfig}
\usepackage{cite}
\usepackage{amsthm}
\usepackage{subfigure}
\definecolor{c1}{rgb}{0,0,0}
\definecolor{c2}{rgb}{0,0,0}

\usepackage{xcolor}

\newsavebox\MBox

\graphicspath{{./Figures/}}

\newtheorem{theorem}{Theorem}

\definecolor{C1}{rgb}{0,0,0}
\newcommand{\X}{\boldsymbol{X}} 
\newcommand{\Y}{\boldsymbol{Y}} 
\begin{document}

\IEEEoverridecommandlockouts
%\title{Alternating Connectivity for a Special Class of Three User Interference Channel}
\title{Topological Interference Management with Alternating Connectivity: The Wyner-Type Three User Interference Channel}
\author{
\IEEEauthorblockN{Soheyl Gherekhloo, Anas Chaaban, and Aydin Sezgin}
\IEEEauthorblockA{Chair of Communication Systems, 
RUB, Germany\\
Email: { \{soheyl.gherekhloo, anas.chaaban, aydin.sezgin\}@rub.de}}
%\thanks{This work is supported in part by the German Research Foundation, Deutsche
%Forschungsgemeinschaft (DFG), Germany, under grant SE 1697/3.
%Anas Chaaban and Aydin Sezgin are with the chair of digital communication systems, Ruhr-Universit\"at Bochum, 44780 Bochum, Germany, Email: anas.chaaban@rub.de, aydin.sezgin@rub.de. Daniela Tuninetti is with the University of Illinois at Chicago, Chicago, IL 60607 USA, Email: danielat@uic.edu.
%}
}

\maketitle
\vspace{-1cm}
\begin{abstract}
Interference management \textcolor{c2}{in a three-user} interference \textcolor{c2}{channel} with alternating connectivity with only topological knowledge at the transmitters is considered. 
The network has a Wyner-type channel flavor, i.e., for each connectivity state the receivers observe \textcolor{C1}{at most} one interference signal in addition to their desired signal. \textcolor{c2}{Degrees of freedom (DoF)} upper bounds and \textcolor{c2}{lower bounds} are derived.
The \textcolor{c2}{lower bounds are obtained from a scheme} based on joint encoding across the alternating states. Given a uniform distribution among the connectivity states,  \textcolor{C1}{it is shown that the channel has} $2+\nicefrac{1}{9}$ \textcolor{C1}{DoF}. 
\textcolor{C1}{This provides an increase in the DoF as compared to encoding over each state separately, which achieves $2$ DoF only.}
\end{abstract}
\section{Introduction}
The smart management of interference beyond the classical approaches of avoidance and suppression is nowadays the focus of research on wireless networks.  The means to apply smart management depend certainly (among other things) on the information \textcolor{c2}{available at the transmitting nodes}, such as channel states.  Often it is assumed that comprehensive channel state information is available at \textcolor{C1}{the} transmitters (CSIT). However, providing comprehensive (or perfect) CSIT is a challenging issue in wireless networks, especially for \textcolor{c2}{networks} with high mobility and size. It is thus of interest to study networks based on the assumption of limited or imperfect CSIT. 

The case of completely stale CSIT (\textcolor{c2}{using the so-called} retrospective interference alignment (IA)) was considered in~\cite{MaddahaliTse} for the broadcast channel with two antennas at the base station and single-antennas at the users. It was shown that a degrees of freedom (DoF) of $\nicefrac{4}{3}$ are achievable. Note that this is less than the DoF of $2$ in the perfect CSIT case, however, more than the DoF of $1$ in the case of completely absent CSIT. The approach was generalized to other networks in~\cite{MalekiJafarShamai}. 
Naturally, it might occur that a mixture of CSIT \textcolor{c2}{quality} is available at the transmitters. This issue was addressed in~\cite{GouJafar_mixedCSIT} and \cite{ShengKobayashiGesbertYi} in which the DoF is studied under the assumption of delayed as well as imperfect current CSIT. As most wireless networks are rather heterogeneous in terms of node mobility and capability, the CSI quality at the transmitters is not the same for all users. This was considered in~\cite{TandonJafarShamaiPoor}, in which users have either perfect, delayed, or no CSIT at all. 

A paradigm shift towards interference management with minimal CSIT has been pursued in \cite{Jafar}. The main assumption of \cite{Jafar} is \textcolor{C1}{restricting the CSI feedback to 1 bit only}; which provides information about presence \textcolor{C1}{or absence} of a link. \textcolor{C1}{A link is assumed to be absent if its corresponding interference \textcolor{c2}{to} noise ratio (INR) is lower than 1.} Clearly, by this assumption the CSIT cannot exceed the topology of the network. Therefore, this problem is called ``topological interference management". It is shown in \cite{Jafar} that the \textcolor{c2}{``topological"} interference management problem for the linear wired and wireless network reduces to a single problem. In other words, solving one of these problems leads to \textcolor{C1}{the solution} for the other one, in such a way that the DoF \textcolor{C1}{of} a linear wireless \textcolor{C1}{network leads to} the capacity of the corresponding linear wired channel, \textcolor{c2}{or vice versa}. 

Note that in~\cite{Jafar} the channels are assumed to be time-invariant, which leads to a fixed connectivity within the network. The extension to time-variant channels and thus to alternating connectivity was considered for the two-user interference channel in~\cite{SunGengJafar}. It was shown that the capacity can only be achieved by jointly encoding across \textcolor{C1}{alternating topologies.} 

In this work, we characterize the DoF of a three user interference channel in which each receiving node is either free of interference or is interfered solely by one transmitter.  The analysis is focused on the corresponding wired network with equiprobable topologies, for which the capacity is characterized. This capacity characterization of the wired network leads then (as mentioned before) to the DoF characterization of the wireless network. 

\section{Motivation}
Consider three adjacent cells \textcolor{C1}{in a wireless network}. In each cell, a base station \textcolor{C1}{wants to send a message to} one \textcolor{C1}{desired} receiver. Suppose that a signal is received under the noise level if the distance between the transmitter (Tx) and the receiver (Rx) is less than the radius of the cell. Therefore, all receivers receive their desired signal over the noise level. However, there are some cases in which the receivers observe one interference signal over the noise level in addition to their desired signal.
\begin{figure} 
\centering
\includegraphics[scale=0.11]{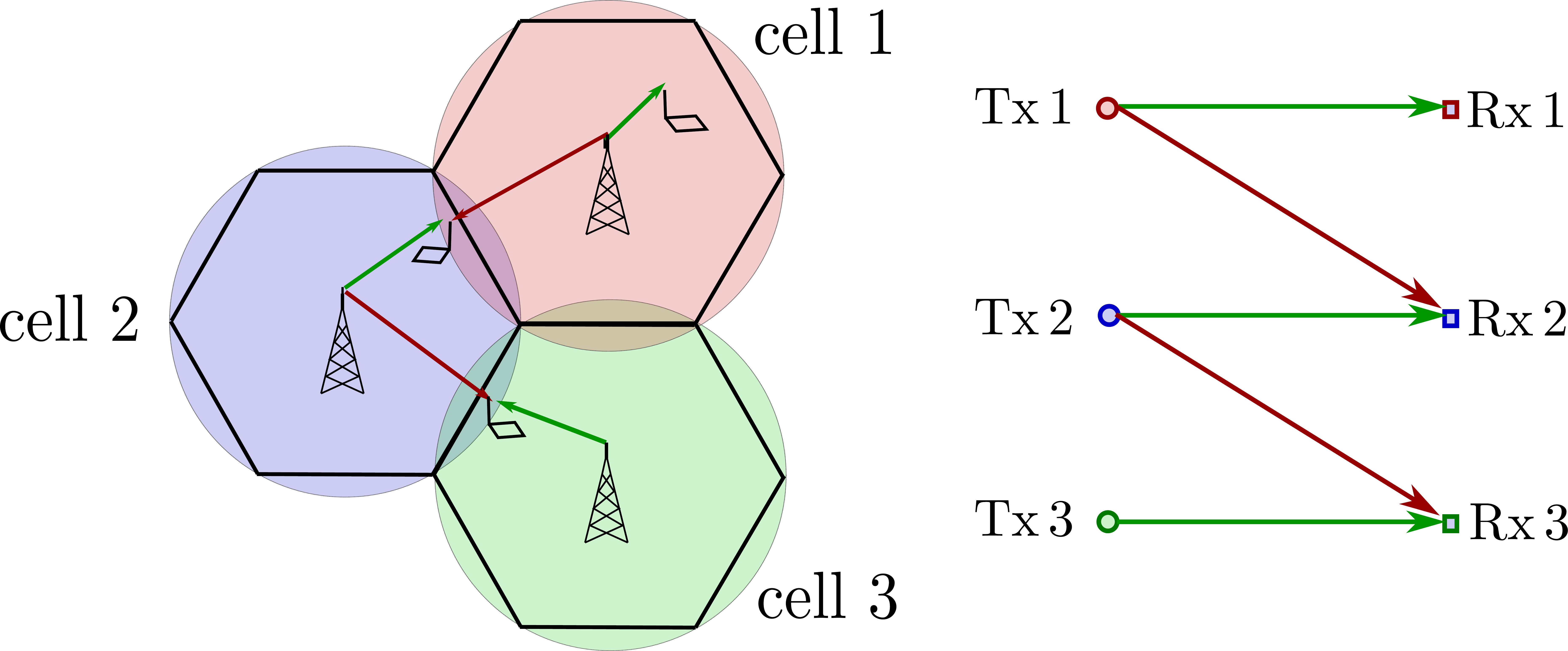}
\caption{\textcolor{C1}{Each} base station \textcolor{C1}{serves} the users located in its cell. However, its signal can be received over the noise level in some areas of the adjacent cells. For example, the users in cell 2 and 3 receive an interference signal from the base stations in cell 1 and 2, respectively. On the right, the topology of this network is shown. \textcolor{C1}{Note that each receiver experiences at most one interference.}}
\label{fig:Hexagonal_cell}
\end{figure}
\textcolor{c2}{This can be seen in} Fig.~\ref{fig:Hexagonal_cell} \textcolor{c2}{which} shows the circular coverage area of three adjacent cells. The area which is allocated to a base station is shown as a hexagon inside a circle. Therefore, there are some areas close to the edges of each cell in which the receiver experiences one interference signal in addition to its own desired signal. As an example, Rx\,2 and Rx\,3 in Fig.~\ref{fig:Hexagonal_cell} observe an interference from undesired base stations Tx\,1 and Tx\,2, respectively. Since an interferer which is weaker than noise does not have an impact on the DoF of the network; the corresponding link to that interferer is assumed to be absent in the topology of the network (see the topology of the wireless network in Fig.~\ref{fig:Hexagonal_cell}).

%By assuming a time variant channel, the topology of the network changes during the transmission. Since in this scenario each receiver observes at most one interference signal in addition to its desired signal, the topology of the network can have 27 states in total.
% Suppose that the receivers start decoding first after the transmission phase. By assuming an infinite memory for the receivers and infinitely long transmission phase, we conclude that the order of occurrence of the states are not important for choosing the optimum transmission scheme but only their probabilities. The main goal of this work is to characterize the DoF of this network with alternating topology for the case when all states occur with the same probability.
\section{System Model}
As it is shown in \cite{Jafar}, the capacity of a wired network normalized by the capacity of a single link gives us the degrees of freedom for the corresponding wireless network. \textcolor{C1}{For simplicity, and in order to avoid the unnecessary treatment of noise in the wireless network which does not have an impact on the DoF of the network, we study the wired noiseless network.}
%As we \textcolor{C1}{are going to} solve the problem \textcolor{C1}{for the wired case,} we define the system model for the wired \textcolor{C1}{network}. 
Consider three Tx\textcolor{c2}{'s} which want to communicate with their desired Rx\textcolor{c2}{'s}. Tx\,$i$, $i\in\{1,2,3\}$ wants to send a message $W_i$ to Rx\,$i$. It encodes this message into a length-$n$ sequence $\X_i= (X_i(1),\ldots,X_i(n))$ \textcolor{C1}{and sends this sequence}. The received symbol at Rx\,${j}$ in $k$th channel use is given by 
\begin{align}
Y_j(k) = \sum_{i=1}^3 h_{ji}(k) X_i(k), \quad \forall j\in\lbrace 1, 2, 3 \rbrace \label{eq:received_symbol}
\end{align}
where $X_i(k)$ and $h_{ji}(k)$ denote the transmitted symbol by Tx\,$i$ and the channel coefficient corresponding to the link between Tx\,$i$ and Rx\,${j}$. All symbols are chosen from a Galois Field $\mathbb{GF}$. Moreover, the linear operations are \textcolor{c2}{performed} over this $\mathbb{GF}$. The capacity of \textcolor{C1}{each channel} is $\log|\mathbb{GF}|$, where $|\mathbb{GF}|$ represents the cardinality of $\mathbb{GF}$. Therefore, \textcolor{C1}{only} one symbol can be transmitted over a link \textcolor{C1}{per} channel use. 

In our model, CSIT is restricted only to the topology of the network. Therefore, the only information available \textcolor{C1}{to} the transmitters is about \textcolor{C1}{the} presence \textcolor{C1}{or absence} of links but not about the channel coefficients. However, \textcolor{C1}{both} the local channel coefficients and the topology of the network \textcolor{C1}{are} known at the receivers. 

Since the channel coefficients change, the topology of the network varies during the transmission. % Therefore, the network can be in different states with some probability.
%Similar to \cite{SunGengJafar}, it is assumed that the desired channels always exist. \textcolor{C1}
{Following the motivation in Fig. \ref{fig:Hexagonal_cell}, the desired channels always exist and each receiver} is \textcolor{C1}{disturbed} by \textcolor{C1}{at most} one interferer. \textcolor{C1}{Therefore, the network has a total of 27 topologies as shown in Fig.~\ref{fig:All_cases}.}

It is worth to note that the receivers have an infinite memory and they start the decoding after receiving \textcolor{C1}{a complete sequence} $\Y_j$. Therefore, the order of the occurrence of the states is not important. \textcolor{c2}{Let $\mathcal{A}$ be a set of states shown in Fig. \ref{fig:All_cases}} and $\boldsymbol{X}_{i,\mathcal{A}}$ be the sequence of transmitted symbols by Tx\,$i$ in all states in $\mathcal{A}$. Assuming a length-$n$ sequence $\X$, the length of $\boldsymbol{X}_{i,\mathcal{A}}$ is $n \lambda_\mathcal{A}$, where $\lambda_\mathcal{A}$ denotes the sum of the probabilities of the states in $\mathcal{A}$.

The goal of this work is to \textcolor{C1}{study} the DoF \textcolor{C1}{gain obtained} by jointly encoding across the alternating topologies, \textcolor{C1}{when all states occur with the same probability.}
\newcommand{\startnodes}[1]{
\node at (0.5,0) [above] {#1};
%\node (t1) at (0,0) [inner sep=0,left]{b1};
\node (t1) at (0,0) [inner sep=0] {};
\node (t2) at (0,-0.4) [inner sep=0] {};
\node (t3) at (0,-0.8) [inner sep=0] {};
\node (r1) at (1,0) [inner sep=0] {};
\node (r2) at (1,-0.4) [inner sep=0] {};
\node (r3) at (1,-0.8) [inner sep=0] {};
\draw[->] (t1) to (r1);
\draw[->] (t2) to (r2);
\draw[->] (t3) to (r3);}
\newcommand{\redpath}[2]{\draw[->,red] (#1) to (#2);}
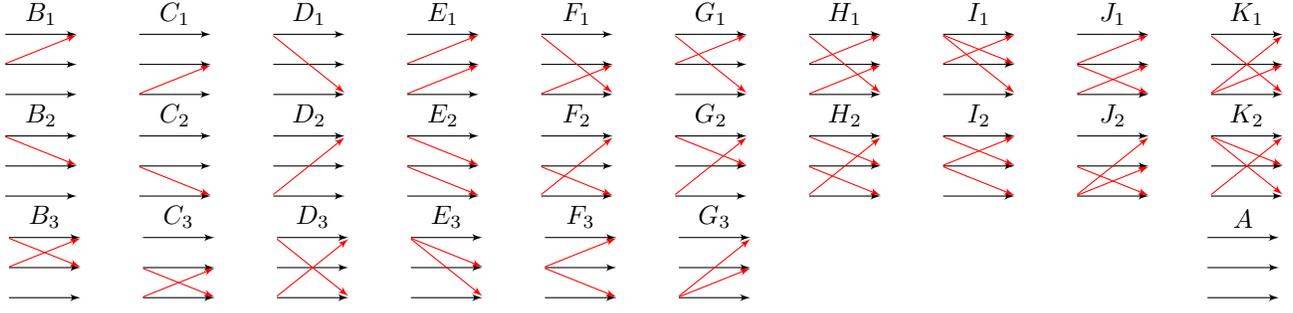
\begin{figure*} 
\centering
%\begin{tikzpicture}
\begin{tikzpicture}
\startnodes{$B_1$}
\redpath{t2}{r1}
\end{tikzpicture}
\hspace{5mm}
\begin{tikzpicture}
\startnodes{$C_1$}
\redpath{t3}{r2}
\end{tikzpicture}
\hspace{5mm}
\begin{tikzpicture}
\startnodes{$D_1$}
\redpath{t1}{r3}
\end{tikzpicture}
\hspace{5mm}
\begin{tikzpicture}
\startnodes{$E_1$}
\redpath{t2}{r1}
\redpath{t3}{r2}
\end{tikzpicture}
\hspace{5mm}
\begin{tikzpicture}
\startnodes{$F_1$}
\redpath{t1}{r3}
\redpath{t3}{r2}
\end{tikzpicture}
\hspace{5mm}
\begin{tikzpicture}
\startnodes{$G_1$}
\redpath{t2}{r1}
\redpath{t1}{r3}
\end{tikzpicture}
\hspace{5mm}	
\begin{tikzpicture}
\startnodes{$H_1$}
\redpath{t2}{r1}
\redpath{t3}{r2}
\redpath{t1}{r3}
\end{tikzpicture}
\hspace{5mm}
\begin{tikzpicture}
\startnodes{$I_1$}
\redpath{t1}{r3}
\redpath{t2}{r1}
\redpath{t1}{r2}
\end{tikzpicture}
\hspace{5mm}
\begin{tikzpicture}
\startnodes{$J_1$}
\redpath{t2}{r3}
\redpath{t2}{r1}
\redpath{t3}{r2}
\end{tikzpicture}
\hspace{5mm}
\begin{tikzpicture}
\startnodes{$K_1$}
\redpath{t1}{r3}
\redpath{t3}{r1}
\redpath{t3}{r2}
\end{tikzpicture}
%\end{tikzpicture}
% next row 

%\vspace{0.1cm}
%\begin{tikzpicture}
\begin{tikzpicture}
\startnodes{$B_2$}
\redpath{t1}{r2}
\end{tikzpicture}
\hspace{5mm}
\begin{tikzpicture}
\startnodes{$C_2$}
\redpath{t2}{r3}
\end{tikzpicture}
\hspace{5mm}
\begin{tikzpicture}
\startnodes{$D_2$}
\redpath{t3}{r1}
\end{tikzpicture}
\hspace{5mm}
\begin{tikzpicture}
\startnodes{$E_2$}
\redpath{t1}{r2}
\redpath{t2}{r3}
\end{tikzpicture}
\hspace{5mm}
\begin{tikzpicture}
\startnodes{$F_2$}
\redpath{t3}{r1}
\redpath{t2}{r3}
\end{tikzpicture}
\hspace{5mm}
\begin{tikzpicture}
\startnodes{$G_2$}
\redpath{t1}{r2}
\redpath{t3}{r1}
\end{tikzpicture}
\hspace{5mm}
\begin{tikzpicture}
\startnodes{$H_2$}
\redpath{t1}{r2}
\redpath{t2}{r3}
\redpath{t3}{r1}
\end{tikzpicture}
\hspace{5mm}
\begin{tikzpicture}
\startnodes{$I_2$}
\redpath{t2}{r3}
\redpath{t2}{r1}
\redpath{t1}{r2}
\end{tikzpicture}
\hspace{5mm}
\begin{tikzpicture}
\startnodes{$J_2$}
\redpath{t2}{r3}
\redpath{t3}{r1}
\redpath{t3}{r2}
\end{tikzpicture}
\hspace{5mm}
\begin{tikzpicture}
\startnodes{$K_2$}
\redpath{t1}{r3}
\redpath{t3}{r1}
\redpath{t1}{r2}
\end{tikzpicture}
%\end{tikzpicture}

%\vspace{0.1cm}
\begin{tikzpicture}
\startnodes{$B_3$}
\redpath{t2}{r1}
\redpath{t1}{r2}
\end{tikzpicture}
\hspace{5mm}
\begin{tikzpicture}
\startnodes{$C_3$}
\redpath{t2}{r3}
\redpath{t3}{r2}
\end{tikzpicture}
\hspace{5mm}
\begin{tikzpicture}
\startnodes{$D_3$}
\redpath{t1}{r3}
\redpath{t3}{r1}
\end{tikzpicture}
\hspace{5mm}
\begin{tikzpicture}
\startnodes{$E_3$}
\redpath{t1}{r2}
\redpath{t1}{r3}
\end{tikzpicture}
\hspace{5mm}
\begin{tikzpicture}
\startnodes{$F_3$}
\redpath{t2}{r1}
\redpath{t2}{r3}
\end{tikzpicture}
\hspace{5mm}
\begin{tikzpicture}
\startnodes{$G_3$}
\redpath{t3}{r1}
\redpath{t3}{r2}
\end{tikzpicture}
\hspace{5mm}
\begin{tikzpicture}
%\startnodes{}
%\redpath{t3}{r1}
%\redpath{t3}{r2}
\end{tikzpicture}
\hspace{5cm}
\begin{tikzpicture}
\startnodes{$A$}
%\redpath{t3}{r1}
%\redpath{t3}{r2}
\end{tikzpicture}
\caption{All possible states for the three users interference channel, when each receiver observes at most one interferer. The desired links are always present.}
\label{fig:All_cases}
\end{figure*}
%\begin{figure*} 
%\centering
%\includegraphics[scale=0.8]{Figures/All_cases.pdf}
%\caption{All possible states for the three users interference channel, when each receiver observes maximum one interference and the desired links are present always.}
%\label{fig:All_cases}
%\end{figure*}
\section{Main Result}
\label{sec:Main_Result}
The following theorem provides the main result of this work.
\begin{theorem}
The three user interference channel with alternating connectivity and equiprobable states with at most one interferer per receiver has DoF=$2+1/9$.
\label{Theorem_Alternating_capacity}
\end{theorem}	
\begin{proof}
We establish Theorem~\ref{Theorem_Alternating_capacity} by showing that the sum capacity of the corresponding wired network is $(2+\frac{1}{9}) \log|\mathbb{GF}|$. In order to do this, we need to find an optimal achievability scheme. The optimality of the scheme is shown by comparing it with a \textcolor{C1}{tight} upper bound of the sum capacity. \textcolor{c2}{We start by proposing an achievability scheme leading to a DoF lower bound denoted \underline{DoF}.}
\subsection*{Achievability:}
The achievability is based on the joint encoding over the sates \cite{SunGengJafar}. To this end, consider states $B_1$, $C_1$, $D_1$, and $H_1$ \textcolor{c2}{in Fig.~\ref{fig:All_cases}}. It can be seen that all interference links in states $B_1$, $C_1$, and $D_1$ are present in state $H_1$. Therefore, we can utilize state $H_1$ to \textcolor{C1}{resolve} the interferences in these states. % new way to explain the jointly encoding more briefly
As it is shown in Fig. \ref{fig:Achievability1}, the symbols $b_1$, $c_2$, and $d_3$ cannot be decoded at the desired receivers in states $B_1$, $C_1$, and $D_1$. 
However, by using the state $H_1$, the transmitters provide the symbols which cause interference in states $B_1$, $C_1$, and $D_1$ to the receivers.
%Therefore, we use state $H_1$ to retransmit the symbols which cause interference in states $B_1$, $C_1$, and $D_1$. Note that these symbols are received interference free at the desired receivers in states $B_1$, $C_1$, and $D_1$. Therefore, each receiver can decode its interfering symbol in state $H_1$. Then, the receivers remove the interference caused by these symbols in $B_1$, $C_1$, and $D_1$. 
Therefore, in total $9$ symbols are decoded correctly at the desired receivers by combining these four states. Similarly, the same joint encoding scheme can be used for $B_2$, $C_2$, $D_2$, and $H_2$ due to symmetry. The remaining states are encoded individually. In all these states except in state $A$, we achieve \underline{DoF}=2 by choosing two active transmitters. For instance, in state $I_1$, \underline{DoF}=2 is achievable when Tx\,2 and Tx\,3 send while Tx\,1 is silent. Overall, the following \underline{DoF} is achievable
\begin{align}
\text{\underline{DoF}}= 
\begin{cases}
9/4 & \quad \text{for }B_1 \cup C_1 \cup D_1 \cup  H_1\\
9/4 & \quad \text{for }B_2 \cup C_2 \cup D_2 \cup  H_2\\
3 & \quad \text{for } A \\
2 & \quad \text{in all remaining 18 states }  	
\end{cases}
\notag
\end{align}
Since, all states occur with equal probability, we can transmit 57 symbols reliably in 27 channel uses in average.  Since every symbol is chosen from $\mathbb{GF}$ with the entropy $\log|\mathbb{GF}|$, the achievable sum rate is 
\begin{align} 
R_\Sigma \leq \left(2+\frac{1}{9}\right)\log|\mathbb{GF}|. \label{achieavablerate}
\end{align}
\newcommand{\achievability}[7]{
\node at (0.5,0) [above] {#1};
\node (t1) at (0,0) [inner sep=0]{};
\node (t1s) at (0,0) [inner sep=0,left]{#2};
\node (t2) at (0,-0.4) [inner sep=0]{};
\node (t2s) at (0,-0.4) [inner sep=0,left]{#3};
\node (t3) at (0,-0.8) [inner sep=0] {};
\node (t3s) at (0,-0.8) [inner sep=0,left] {#4};
\node (r1) at (1,0) [inner sep=0] {};
\node (r1s) at (1,0) [inner sep=0,right] {#5};
\node (r2) at (1,-0.4) [inner sep=0] {};
\node (r2s) at (1,-0.4) [inner sep=0,right] {#6};
\node (r3) at (1,-0.8) [inner sep=0] {};
\node (r3s) at (1,-0.8) [inner sep=0,right] {#7};
\draw[->] (t1) to (r1);
\draw[->] (t2) to (r2);
\draw[->] (t3) to (r3);}

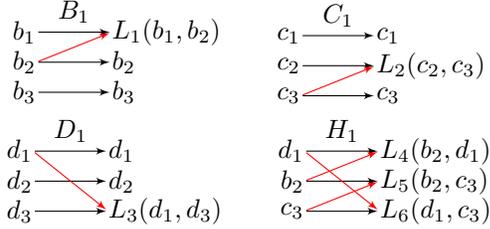
\begin{figure} 
\centering
\begin{tikzpicture}
\achievability{$B_1$}{$b_1$}{$b_2$}{$b_3$}{$L_1(b_1,b_2)$}{$b_2$}{$b_3$}
\redpath{t2}{r1}
\end{tikzpicture}
\hspace{5mm}
\begin{tikzpicture}
\achievability{$C_1$}{$c_1$}{$c_2$}{$c_3$}{$c_1$}{$L_2(c_2,c_3)$}{$c_3$}
\redpath{t3}{r2}
\end{tikzpicture}

%\end{tikzpicture}

% next row 
\vspace{0.1cm}
%\begin{tikzpicture}

\begin{tikzpicture}
\achievability{$D_1$}{$d_1$}{$d_2$}{$d_3$}{$d_1$}{$d_2$}{$L_3(d_1,d_3)$}
\redpath{t1}{r3}
\end{tikzpicture}
\hspace{5mm}
\begin{tikzpicture}
\achievability{$H_1$}{$d_1$}{$b_2$}{$c_3$}{$L_4(b_2,d_1)$}{$L_5(b_2,c_3)$}{$L_6(d_1,c_3	)$}
\redpath{t1}{r3}
\redpath{t2}{r1}
\redpath{t3}{r2}
\end{tikzpicture}
\caption{By combining these four states, we can recover $9$ symbols. However, by considering them separately, we cannot exceed $15/2$ symbols.}
\label{fig:Achievability1}
\end{figure}
%\begin{figure} 
%\centering
%\includegraphics[scale=0.4]{Figures/Achievability1.pdf}
%\caption{By combining these four states together, we can transmit $9$ symbols. However, by considering them separately, we cannot exceed $8.5$ symbols.}
%\label{fig:Achievability1}
%\end{figure}
%Note that there are different transmission schemes in order to achieve the rate given in \eqref{achieavablerate}. All these schemes will be required in the non-equiprobable case. Since studying all these schemes is beyond the goal of this work, we do not discuss the other ones.
\subsection*{Upper bound:}
%In order to show that the achievability scheme is optimal, we need to compare it with \textcolor{C1}{a tight} upper bound. 
We establish the upper bound as follows
%\allowdisplaybreaks
\begin{align}
n R_\Sigma =&  \sum_{i=1}^{3} H(W_i) \notag \\
=&  \sum_{i=1}^{3} H(W_i) + H(W_i|\Y_i) - H(W_i|\Y_i) \notag \\
\overset{(a)}{\leq} &  \sum_{i=1}^{3} I(W_i;\Y_i) + 3n\epsilon_n, \label{eq:R_sum} 
%\leq &  I(W_1;\Y_1,W_2,W_3) + I(W_2;\Y_2,W_1,W_3) \label{eq:2R_sum}\\ & + I(W_3;\Y_3,W_1,W_2) 
% +\sum_{i=1}^{3}I(W_i;\Y_i) + 6n\epsilon_n.\notag
\end{align}
where ($a$) follows from Fano's inequality \textcolor{c2}{and $\epsilon_n\rightarrow 0$ when $n\rightarrow \infty$}. By multiplying the inequality in \eqref{eq:R_sum} by 2, every mutual information appears twice which corresponds to creating (virtually) three additional receivers. In the next step, we give side information to the actual receivers. The side information equals to the undesired messages at those receivers. Therefore, we write 
\begin{align}
2n R_\Sigma \leq &  I(W_1;\Y_1,W_2,W_3) + I(W_2;\Y_2,W_1,W_3) \label{eq:2R_sum}\\ & + I(W_3;\Y_3,W_1,W_2) 
 +\sum_{i=1}^{3}I(W_i;\Y_i) + 6n\epsilon_n.\notag
\end{align}
\iffalse

In ($b$), we doubled the sum rate in order to cancel some positive entropy terms from one sum rate with the corresponding negative terms in the other one. Since in this scenario the receivers observe at most one interference signal, the negative entropy terms in the mutual information, i.e. $-H(\Y_i|W_i)$, are equal to the entropy of a set of transmitted symbols. Providing all undesired messages to the receivers as side information in \eqref{eq:2R_sum} is equivalent to remove all interference links. In this way we generate the entropy of the transmitted symbols. The goal is now to cancel the entropy of the transmitted symbols with the corresponding negative terms.

\fi
By using the chain rule and since the messages of three transmitters are independent from each other, we write
\begin{align}
2nR_\Sigma &\leq I(W_1;\Y_1|W_2,W_3) + I(W_2;\Y_2|W_1,W_3) \notag\\ &\quad+ I(W_3;\Y_3|W_1,W_2) 
+ \sum_{i=1}^{3}I(W_i;\Y_i) + 6n\epsilon_n.  \label{eq:Al_4_1_1}
\end{align}
By expressing the mutual information as entropy terms, \eqref{eq:Al_4_1_1} is restated as
\begin{align}
2nR_\Sigma &\leq  H(\Y_1|W_2,W_3) - H(\Y_1|W_2,W_3,W_1) \notag \\ & \quad+ H(\Y_2|W_1,W_3) -  H(\Y_2|W_1,W_3,W_2) \notag\\ &\quad+ H(\Y_3|W_1,W_2) - H(\Y_3|W_1,W_2,W_3)
\notag\\ & \quad + \sum_{i=1}^{3}I(W_i;\Y_i) + 6n\epsilon_n. \label{eq:Al_4_1} 
\end{align}
Note that knowing all messages, $\Y_i$ can be reconstructed. Therefore, $H(\Y_i|W_1,W_2,W_3) = 0$. The first term in \eqref{eq:Al_4_1} reduces to
\begin{align}
H(\Y_1|W_2,W_3) = H(\X_1), \notag
\end{align}
as $\X_1$ is independent of $W_2$ and $W_3$ and the fact that \textcolor{c2}{scaling} a discrete random variable \textcolor{c2}{by} a constant does not influence entropy \cite{CoverThomas}. Similar \textcolor{c2}{treatment applies to} $H(\Y_2|W_1,W_3)$ and $H(\Y_3|W_1,W_2)$ in \eqref{eq:Al_4_1}. Next, we rewrite \eqref{eq:Al_4_1}  as shown in \eqref{eq:Al_4_2} on the top of next page. The parameters $\Delta_i$, $\Gamma_i$, and $\Theta_i$, $i\in\{1,2,3\}$ are defined as follows
%The first three non-zero entropy terms in \eqref{eq:Al_4_1} can be replaced by the entropy of the desired transmitted symbols multiplied by the corresponding channel coefficients. \textcolor{C1}{Multiplying a discrete random variable with a constant does not have any impact on the probability mass function and} on the \textcolor{C1}{discrete} entropy. Therefore, we drop the channel coefficients and rewrite \eqref{eq:Al_4_1} as shown in \eqref{eq:Al_4_2} on the top of next page, where $\Delta_i$, $\Gamma_i$, and $\Theta_i$, $i\in\{1,2,3\}$ are the sets of different states which are defined as follows.
\begin{align}
\Delta_1&=\lbrace D_1,F_1,G_1,H_1,I_1,K_1,K_2,D_3\rbrace \notag\\
\Gamma_1&=\lbrace B_2,E_2,G_2,H_2,I_2,B_3 \rbrace \notag\\
\Theta_1&=\overline{\lbrace E_3\rbrace \cup \Delta_1 \cup \Gamma_1}\notag\\
\Delta_2&=\lbrace B_1,E_1,G_1,H_1,I_1,J_1,I_2,B_3\rbrace \notag\\
\Gamma_2&=\lbrace C_2,E_2,F_2,H_2,J_2,C_3\rbrace \notag\\
\Theta_2&=\overline{\lbrace F_3\rbrace \cup \Delta_2\cup \Gamma_2 }\notag\\
\Delta_3&=\lbrace C_1,E_1,F_1,H_1,J_1,K_1,J_2,C_3\rbrace \notag\\
\Gamma_3&=\lbrace D_2,F_2,G_2,H_2,K_2,D_3\rbrace \notag\\
\Theta_3&= \overline{\lbrace G_3\rbrace \cup \Delta_3\cup \Gamma_3}. \notag
\end{align}
The notation $\overline{\mathcal{A}}$ denotes the complement set of $\mathcal{A}$. 
\iffalse
Note that $\Delta_i$ and $\Gamma_i$ are all states except $E_3$, $F_3$, and $G_3$ in which Tx\,$i$ causes interference. 
\fi
%In general, $\Delta_i$ and $\Gamma_i$ are the set of states in which Tx\,$i$ causes interference to Rx\,$j$ where $j=[(i-1)+2\mod 3]+1$ and $j=[(i-1)+1\mod 3]+1$, respectively.
%Note that by considering only a subset of all 27 states, the length of the sequence must be also taken into account. Considering a set of some states, e.g. $\mathcal{A}$, the sequence of transmitted symbols by Tx\,i in all states in $\mathcal{A}$ is denoted by $\X_{i\mathcal{A}}$. \textcolor{C1}{Note that this notation shows us the length of the sequence implicitly by the set of states which is represented by the sequence.} For example, the length of the sequence $\X_{i\mathcal{A}}$ is $n \lambda_{\mathcal{A}}$, where $\lambda_{\mathcal{A}}$ denotes the sum of the probabilities of the states in $\mathcal{A}$. 
\begin{figure*}
\begin{align}
%2nR_\Sigma \leq &  H(\X_{1,E_3},\X_{1,\Delta_1},\X_{1,\Gamma_1},\X_{1,\Theta_1 }|W_2,W_3) 
% + H(\X_{2,F_3},\X_{2,\Delta_2},\X_{2,\Gamma_2},\X_{2,\Theta_2 }|W_1,W_3) \notag \\ 
% &+  H(\X_{3,G_3},\X_{3,\Delta_3},\X_{3,\Gamma_3},\X_{3,\Theta_3 }|W_1,W_2) +  H(\X_{1,E_3},\Y_{1,\overline{E_3}}) - H(\X_{2,F_3},\X_{2,\Delta_2},\X_{3,G_3},\X_{3,\Gamma_3},\X_{3,K_1},\X_{3,J_2}) \notag\\
% & + H(\X_{2,F_3},\Y_{2,\overline{F_3}}) - H(\X_{3,G_3},\X_{3,\Delta_3},\X_{1,E_3},\X_{1,\Gamma_1},\X_{1,I_1},\X_{1,K_2}) \notag\\ 
% & + H(\X_{3,G_3},\Y_{3,\overline{G_3}}) - H(\X_{1,E_3},\X_{1,\Delta_1},\X_{2,F_3},\X_{2,\Gamma_2},\X_{2,J_1},\X_{2,I_2}) + 6n\epsilon_n \label{eq:Al_4_2} 
2nR_\Sigma \leq &  H(\X_{1,E_3},\X_{1,\Delta_1},\X_{1,\Gamma_1},\X_{1,\Theta_1 })
 + H(\X_{2,F_3},\X_{2,\Delta_2},\X_{2,\Gamma_2},\X_{2,\Theta_2 }) +  H(\X_{3,G_3},\X_{3,\Delta_3},\X_{3,\Gamma_3},\X_{3,\Theta_3 }) \notag \\ 
  &+  H(\X_{1,E_3},\Y_{1,\overline{E_3}}) - H(\X_{2,F_3},\X_{2,\Delta_2},\X_{3,G_3},\X_{3,\Gamma_3},\X_{3,K_1},\X_{3,J_2}) \notag\\
 & + H(\X_{2,F_3},\Y_{2,\overline{F_3}}) - H(\X_{3,G_3},\X_{3,\Delta_3},\X_{1,E_3},\X_{1,\Gamma_1},\X_{1,I_1},\X_{1,K_2}) \notag\\ 
 & + H(\X_{3,G_3},\Y_{3,\overline{G_3}}) - H(\X_{1,E_3},\X_{1,\Delta_1},\X_{2,F_3},\X_{2,\Gamma_2},\X_{2,J_1},\X_{2,I_2}) + 6n\epsilon_n \label{eq:Al_4_2} 
% \leq & H(\X_{1,E_3}) + H(\X_{1,\Delta_1}|\X_{1,E_3}) + H(\X_{1,\Gamma_1}|\X_{1,E_3})+ H(\X_{1,\Theta_1 })+  H(\X_{2,F_3}) + H(\X_{2,\Delta_2}|\X_{2,F_3})   \notag \\
% & + H(\X_{2,\Gamma_2}|\X_{2,F_3})+ H(\X_{2,\Theta_2 })+ 
% H(\X_{3,G_3}) + H(\X_{3,\Delta_3}|\X_{3,G_3}) + H(\X_{3,\Gamma_3}|\X_{3,G_3})+ H(\X_{3,\Theta_3 })
% +  H(\X_{1,E_3})    \notag \\
%   &+ H(\Y_{1,\overline{E_3}}|\X_{1,E_3}) - H(\X_{2,F_3})- H(\X_{2,\Delta_2}|\X_{2,F_3}) -H(\X_{3,G_3})- H(\X_{3,\Gamma_3}|\X_{3,G_3})\notag \\  &- H(\X_{3,K_1},\X_{3,J_2}|\X_{3,G_3},\X_{3,\Gamma_3})  + H(\X_{2,F_3}) + H(\Y_{2,\overline{F_3}}|\X_{2,F_3}) - H(\X_{3,G_3})- H(\X_{3,\Delta_3}|\X_{3,G_3}) - H(\X_{1,E_3}) \notag \\ &- H(\X_{1,\Gamma_1}|\X_{1,E_3}) - H(\X_{1,I_1},\X_{1,K_2}|\X_{1,E_3},\X_{1,\Gamma_1}) + H(\X_{3,G_3}) +H(\Y_{3,\overline{G_3}}|\X_{3,G_3})   -H(\X_{1,E_3})\notag \\ &-H(\X_{1,\Delta_1},|\X_{1,E_3}) - H(\X_{2,F_3})- H(\X_{2,\Gamma_2}|\X_{2,F_3}) - H(\X_{2,J_1},\X_{2,I_2}|\X_{2,F_3},\X_{2,\Gamma_2}) + 6n\epsilon_n \label{eq:Al_4_3}
\end{align}
\hrule
%\begin{subequations}
%\begin{align}
%(T_1)  & \leq  H(\X_{1,E_3}) + H(\X_{1,\Delta_1}|\X_{1,E_3}) + H(\X_{1,\Gamma_1}|\X_{1,E_3})+ H(\X_{1,\Theta_1 }) \label{eq:Al_4_3_1} \\
%(T_2) &\leq H(\X_{2,F_3}) + H(\X_{2,\Delta_2}|\X_{2,F_3}) + H(\X_{2,\Gamma_2}|\X_{2,F_3})+ H(\X_{2,\Theta_2 })
%\label{eq:Al_4_3_2} \\
%(T_3) &\leq H(\X_{3,G_3}) + H(\X_{3,\Delta_3}|\X_{3,G_3}) + H(\X_{3,\Gamma_3}|\X_{3,G_3})+ H(\X_{3,\Theta_3 })\label{eq:Al_4_3_3}\\
%(T_4) & \leq H(\X_{1,E_3}) + H(\Y_{1,\overline{E_3}}|\X_{1,E_3})\label{eq:Al_4_3_4} \\
%(T_5) & \leq -H(\X_{2,F_3}) - H(\X_{2,\Delta_2}|\X_{2,F_3}) -H(\X_{3,G_3})- H(\X_{3,\Gamma_3}|\X_{3,G_3})\label{eq:Al_4_3_5} \\
%(T_6) & \leq H(\X_{2,F_3}) + H(\Y_{2,\overline{F_3}}|\X_{2,F_3})\label{eq:Al_4_3_6} \\
%(T_7) & \leq   - H(\X_{3,G_3})- H(\X_{3,\Delta_3}|\X_{3,G_3}) - H(\X_{1,E_3}) - H(\X_{1,\Gamma_1}|\X_{1,E_3})\label{eq:Al_4_3_7} \\
%(T_8) & \leq H(\X_{3,G_3}) +H(\Y_{3,\overline{G_3}}|\X_{3,G_3}) \label{eq:Al_4_3_8} \\
%(T_9) &\leq -H(\X_{1,E_3})-H(\X_{1,\Delta_1}
%|\X_{1,E_3}) - H(\X_{2,F_3})- H(\X_{2,\Gamma_2}|\X_{2,F_3})  \label{eq:Al_4_3_9}
%% - H(\X_{1,I_1},\X_{1,K_2}|\X_{1,E_3},\X_{1,\Gamma_1}) +   -H(\X_{1,E_3})\notag \\ &-H(\X_{1,\Delta_1},|\X_{1,E_3}) - H(\X_{2,F_3})- H(\X_{2,\Gamma_2}|\X_{2,F_3}) - H(\X_{2,J_1},\X_{2,I_2}|\X_{2,F_3},\X_{2,\Gamma_2}) + 6n\epsilon_n \label{eq:Al_4_3}
%\end{align}
\begin{align}
 H(\X_{1,E_3},\X_{1,\Delta_1},\X_{1,\Gamma_1},\X_{1,\Theta_1 })  & \leq  H(\X_{1,E_3}) + H(\X_{1,\Delta_1}|\X_{1,E_3}) + H(\X_{1,\Gamma_1}|\X_{1,E_3})+ H(\X_{1,\Theta_1 }) \label{eq:Al_4_3_1} \\
 H(\X_{2,F_3}, \X_{2,\Delta_2},\X_{2,\Gamma_2},\X_{2,\Theta_2 }) &\leq H(\X_{2,F_3}) + H(\X_{2,\Delta_2}|\X_{2,F_3}) +  H(\X_{2,\Gamma_2}|\X_{2,F_3})+ H(\X_{2,\Theta_2 })
\label{eq:Al_4_3_2} \\
 H(\X_{3,G_3},\X_{3,\Delta_3},\X_{3,\Gamma_3},\X_{3,\Theta_3 }) &\leq H(\X_{3,G_3}) + H(\X_{3,\Delta_3}|\X_{3,G_3}) +  H(\X_{3,\Gamma_3}|\X_{3,G_3})+ H(\X_{3,\Theta_3 })\label{eq:Al_4_3_3} \\
 H(\X_{1,E_3},\Y_{1,\overline{E_3}})  & \leq H(\X_{1,E_3}) + H(\Y_{1,\overline{E_3}}|\X_{1,E_3})\label{eq:Al_4_3_4} \\
 H(\X_{2,F_3},\X_{2,\Delta_2},\X_{3,G_3},\X_{3,\Gamma_3},\X_{3,K_1},\X_{3,J_2}) & \geq H(\X_{2,F_3}) + H(\X_{2,\Delta_2}|\X_{2,F_3}) + H(\X_{3,G_3})+ H(\X_{3,\Gamma_3}|\X_{3,G_3})\label{eq:Al_4_3_5} \\
 H(\X_{2,F_3},\Y_{2,\overline{F_3}})  & \leq H(\X_{2,F_3}) + H(\Y_{2,\overline{F_3}}|\X_{2,F_3})\label{eq:Al_4_3_6} \\
 H(\X_{3,G_3},\X_{3,\Delta_3},\X_{1,E_3},\X_{1,\Gamma_1},\X_{1,I_1},\X_{1,K_2}) & \geq   H(\X_{3,G_3})+ H(\X_{3,\Delta_3}|\X_{3,G_3}) + H(\X_{1,E_3}) + H(\X_{1,\Gamma_1}|\X_{1,E_3})\label{eq:Al_4_3_7} \\
 H(\X_{3,G_3},\Y_{3,\overline{G_3}})  & \leq H(\X_{3,G_3}) +H(\Y_{3,\overline{G_3}}|\X_{3,G_3}) \label{eq:Al_4_3_8} \\
 H(\X_{1,E_3},\X_{1,\Delta_1},\X_{2,F_3},\X_{2,\Gamma_2},\X_{2,J_1},\X_{2,I_2}) &\geq H(\X_{1,E_3})+H(\X_{1,\Delta_1}
|\X_{1,E_3}) + H(\X_{2,F_3})+ H(\X_{2,\Gamma_2}|\X_{2,F_3})  \label{eq:Al_4_3_9}
% - H(\X_{1,I_1},\X_{1,K_2}|\X_{1,E_3},\X_{1,\Gamma_1}) +   -H(\X_{1,E_3})\notag \\ &-H(\X_{1,\Delta_1},|\X_{1,E_3}) - H(\X_{2,F_3})- H(\X_{2,\Gamma_2}|\X_{2,F_3}) - H(\X_{2,J_1},\X_{2,I_2}|\X_{2,F_3},\X_{2,\Gamma_2}) + 6n\epsilon_n \label{eq:Al_4_3}
\end{align}
%\end{subequations}
\hrule
\end{figure*}
By using the chain rule, together with the facts that conditioning does not increase entropy, and that the messages of the users are independent of each other, the individual terms in \eqref{eq:Al_4_2} can be rewritten as in \eqref{eq:Al_4_3_1}-\eqref{eq:Al_4_3_9} \textcolor{c2}{on the top of next page}. 
We can see that \textcolor{c2}{by substituting \eqref{eq:Al_4_3_1}-\eqref{eq:Al_4_3_9} into~\eqref{eq:Al_4_2} many terms will cancel out} and we can rewrite \eqref{eq:Al_4_2} as 
\begin{align}
2n R_\Sigma\leq & \sum_{i=1}^3 H(\X_{i,\Theta_i})+ H(\Y_{1,\overline{E_3}}|\X_{1,E_3}) \label{eq:Al_4_4} \\
&+ H(\Y_{2,\overline{F_3}}|\X_{2,F_3}) + H(\Y_{3,\overline{G_3}}|\X_{3,G_3}) + 6n\epsilon_n.\notag
\end{align}
The inequality \eqref{eq:Al_4_4} \textcolor{c2}{can be further upper bounded by}
\begin{align}
2nR_\Sigma\leq &  \log|\mathbb{GF}|[n\lambda_{\Theta_1}+ n\lambda_{\Theta_2}+ n\lambda_{\Theta_3}  + n(1-\lambda_{E_3})  \notag \\
&+ n(1-\lambda_{F_3})+ n(1-\lambda_{G_3})] + 6n\epsilon_n, \label{eq:Al_4_5}
\end{align}
where we used the chain rule, the fact that conditioning does not increase the entropy, and that the entropy of discrete random variable in $\mathbb{GF}$ is upper bounded by $\log|\mathbb{GF}|$ \cite{CoverThomas}.
%We use the chain rule \textcolor{C1}{along with} the definition of sets $\Theta_{i}$. Then, we drop the conditions because conditioning does not increase the entropy. \textcolor{C1}{Since} the symbols are chosen from $\mathbb{GF}$, the entropy of the symbols cannot exceed $\log|\mathbb{GF}|$ \cite{CoverThomas}. Therefore, inequality \eqref{eq:Al_4_4} is restated as
%\begin{align}
%2nR_\Sigma\leq &  \log|\mathbb{GF}|[n\lambda_{\Theta 1}+ n\lambda_{\Theta 2}+ n\lambda_{\Theta 3}  + n(1-\lambda_{E3})  \notag \\
%&+ n(1-\lambda_{F3})+ n(1-\lambda_{G3})] + 6n\epsilon_n \label{eq:Al_4_5}
%\end{align}

Since the set $\Theta_i$ consists of 12 states, $\lambda_{\Theta_i} =\frac{12}{27}$ if all states are equiprobable. \textcolor{c2}{Next, we divide the inequality in \eqref{eq:Al_4_5} by $2n$, and let $n\rightarrow \infty$ to obtain} 
\begin{align}
%\frac{R_1+R_2+R_3}{\log|\mathbb{GF}|}\leq & \frac{1}{2}\left[3\frac{12}{27} + 3 \left(1-\frac{1}{27}\right) \right] \notag \\
R_\Sigma \leq & \left(2 + \frac{1}{9}\right) \log|\mathbb{GF}|. \label{tightub}
\end{align}
This agrees with the lower bound in~(\ref{achieavablerate}). 
Normalizing the result by $\log|\mathbb{GF}|$, we get the DoF for the wireless case which proves Theorem \ref{Theorem_Alternating_capacity}.
\end{proof}
%\section{Comparing with separate encoding}
%\section{Discussion}

We observe from Theorem~\ref{Theorem_Alternating_capacity} that no joint processing is necessary for $\overline{\{B_1,C_1,D_1,H_1,B_2,C_2,D_2,H_2\}}$. However, for $\{B_1,C_1,D_1,H_1,B_2,C_2,D_2,H_2\}$, we need joint encoding to achieve the optimal DoF. The alternative approach would be to treat these states separately as well. This would result in a DoF=$3/2$ and  DoF=$2$ for the states $\lbrace H_1,H_2 \rbrace$ (as shown in \cite{ZhouYu}) and $\{B_1,C_1,D_1,B_2,C_2,D_2\}$ (as shown in \cite{EtkinTseWang}), respectively. \textcolor{c2}{Therefore}, the overall DoF=2 is optimal for separate encoding while by using joint encoding across the alternating topologies $2+1/9$ is the optimal achievable DoF.
\section{Conclusion}
\textcolor{c2}{We studied }the DoF \textcolor{c2}{of} the three users interference channel with an alternating connectivity with only topological knowledge at the transmitters. To do this, we proposed a new joint encoding across the alternating topologies. Moreover, a new genie aided upper bound is established to verify the optimality of the joint encoding scheme. The upper bound is tight for the equiprobable case. As future work, the non-equiprobable case will be addressed. However, this extension is non-trivial due to the \textcolor{c2}{increase} in the \textcolor{c2}{number of} possible combination of states.
%However, the extension of this work for the non-equiprobable case will be published also in a longer version.
%\vspace{-.5cm}
%\bibliography{/home/chaaban/Documents/tex/myBib}
\bibliography{myBib}
\end{document}